\documentclass[twocolumn,english,prl, letterpaper,showpacs]{revtex4}
\pdfoutput=1

\usepackage{amsthm}
\usepackage{amsmath}
\usepackage{graphicx}
\usepackage{amssymb}
\usepackage{esint}
\usepackage[unicode=true,pdfusetitle,
 bookmarks=true,bookmarksnumbered=false,bookmarksopen=false,
 breaklinks=false,pdfborder={0 0 1},backref=false,colorlinks=false]
 {hyperref}

\makeatletter
\@ifundefined{textcolor}{}
{%
 \definecolor{BLACK}{gray}{0}
 \definecolor{WHITE}{gray}{1}
 \definecolor{RED}{rgb}{1,0,0}
 \definecolor{GREEN}{rgb}{0,1,0}
 \definecolor{BLUE}{rgb}{0,0,1}
 \definecolor{CYAN}{cmyk}{1,0,0,0}
 \definecolor{MAGENTA}{cmyk}{0,1,0,0}
 \definecolor{YELLOW}{cmyk}{0,0,1,0}
 }
\theoremstyle{plain}
\newtheorem{thm}{Theorem}
\theoremstyle{definition}
\newtheorem{defn}[thm]{Definition}
\theoremstyle{plain}
\newtheorem{prop}[thm]{Proposition}
 \ifx\proof\undefined\
   \newenvironment{proof}[1][\proofname]{\par
     \normalfont\topsep6\p@\@plus6\p@\relax
     \trivlist
     \itemindent\parindent
     \item[\hskip\labelsep
           \scshape
       #1]\ignorespaces
   }{%
     \endtrivlist\@endpefalse
   }
   \providecommand{\proofname}{Proof}
 \fi

\makeatother
\usepackage{tikz}

\makeatother

\begin{document}
\global\long\def\ket#1{|#1\rangle}
\global\long\def\bra#1{\langle#1|}
\global\long\def\id{I}
\global\long\def\proj#1{|#1\rangle\langle#1|}
\global\long\def\ketbra#1#2{|#1\rangle\langle#2|}
\global\long\def\braket#1#2{\langle#1|#2\rangle}
\def\eq#1{Eq.~\eqref{eq:#1}}
\def\fig#1{Fig.~\ref{fig:#1}}
\def\eq#1{Eq.~\eqref{eq:#1}}
\def\II{1\!\mathrm{l}}
\def\cA{\mathcal{A}}
\def\cB{\mathcal{B}}
\def\cC{\mathcal{C}}
\def\cD{\mathcal{D}}
\def\cE{\mathcal{E}}
\def\cF{\mathcal{G}}
\def\cG{\mathcal{G}}
\def\cH{\mathcal{H}}
\def\cI{\mathcal{I}}
\def\cJ{\mathcal{J}}
\def\cK{\mathcal{K}}
\def\cL{\mathcal{L}}
\def\cM{\mathcal{M}}
\def\cN{\mathcal{N}}
\def\cO{\mathcal{O}}
\def\cP{\mathcal{P}}
\def\cQ{\mathcal{Q}}
\def\cR{\mathcal{R}}
\def\cS{\mathcal{S}}
\def\cT{\mathcal{T}}
\def\cU{\mathcal{U}}
\def\cV{\mathcal{V}}
\def\cW{\mathcal{W}}
\def\cX{\mathcal{X}}
\def\cY{\mathcal{Y}}
\def\cZ{\mathcal{Z}}
\def\bL{{\bf L}}
\def\bR{{\bf R}}
\def\Tr{\mathrm{Tr}}

\title{Local topological order inhibits thermal stability in 2D}

\author{Olivier Landon-Cardinal}
\email{olivier.landon-cardinal@usherbrooke.ca}
\affiliation{D\'epartement de Physique, Universit\'e de Sherbrooke, Qu\'ebec, J1K 2R1, Canada}
\author{David Poulin}
\email{david.poulin@usherbrooke.ca}
\affiliation{D\'epartement de Physique, Universit\'e de Sherbrooke, Qu\'ebec, J1K 2R1, Canada}

\pacs{03.67.Pp, 03.65.Ud, 03.67.Ac}
\begin{abstract}
We study the robustness of quantum information stored in the degenerate
ground space of a local, frustration-free Hamiltonian with commuting terms on a 2D spin lattice. On
one hand, a macroscopic energy barrier separating the distinct ground
states under local transformations would protect the information from
thermal fluctuations. On the other hand, local topological order would shield
the ground space from static perturbations. Here we
demonstrate that local topological order implies a constant energy barrier, thus inhibiting thermal stability.
\end{abstract}
\maketitle

A self-correcting quantum memory \cite{Bacon06} is a physical system whose quantum state can be
preserved over a long period of time {\em without} the need for any external
intervention. The archetypical self-correcting classical memory is the
two-dimensional (2D) Ising ferromagnet. The ground state of this system is two-fold
degenerate---all-spin up and all-spin down---so it can store one bit of
information. If the memory is put into contact with a heat bath after being
initialized in one of these ground states, thermal fluctuations will lead to
the creation of small error droplets of inverted spins. The boundary of such
droplets are domain walls, i.e.,  one-dimensional excitations whose energy is
proportional to the droplet perimeter. If the temperature is below the critical
Curie temperature, the Boltzman factor will prevent the creation of macroscopic
error droplets. Thus, when the system is cooled down (either physically or
algorithmically) after some macroscopic storage time, it will very likely return
to its original ground state: the memory is thermally stable. 

This behaviour contrasts with the 1D Ising ferromagnet whose domain
walls are point-like excitations. Therefore, they can freely diffuse on the
chain at no energy cost. As a consequence, arbitrarily large error droplets can
form, so this 1D memory is thermally unstable. 

While the 2D Ising ferromagnet features thermal stability, it is
vulnerable to static, local perturbations. Indeed, an arbitrarily weak magnetic
field breaks the ground state degeneracy and favours one ground state over the
other.
 When this perturbed system is
subject to thermal fluctuations, the bulk
contribution of the magnetic field overwhelms the boundary tension of the domain
wall, so once error droplets reach a critical size, they rapidly
expand to corrupt the memory. This type of instability plagues any systems with a local order parameter, so they cannot be robust quantum memories. Indeed, distinct ground states give different values of this order parameter, so a local field coupling to the order parameter lifts degeneracy.

In 2D and higher, there exists quantum systems with no local order parameter and whose spectrum is
stable under weak, local perturbations. These systems have a degenerate ground
state separated from the other energy levels by a constant energy gap, and
perturbations only alter these features by an exponentially vanishing amount as
a function of the system size.  Kitaev's toric code~\cite{Kitaev03} is the best
known example of this type. However, excitations in Kitaev's code are point-like
objects---as for the 1D Ising model---so it does not offer a
macroscopic energy barrier protection to thermal
fluctuations~\cite{DKL+02,Bacon06,AFH07,AFH09,NO08a}.

In this work,
we study the possibility of combining the thermal stability of the
2D Ising model with the spectral stability of Kitaev's code to
obtain a robust quantum memory in 2D. 
We consider $d$-level spins located at the vertices $V$ of a 2D lattice $\Lambda = (V,E)$, with  Hamiltonian 
\begin{equation}
H = -\sum_{X \subset V} P_X, \ {\rm with}\  [P_X,P_Y] = 0 \ {\rm and} \ \|P_X\| \leq 1.
\label{eq:H}
\end{equation}
We denote the number of spins $N \equiv |V|$.
The term $P_X$ is supported on the subset $X$ of the spins, i.e., it acts trivially on the complement $\overline X = V-X$ of $X$. The Hamiltonian is local in the sense that $P_X = 0$ whenever
$X$ has radius larger than some constant $w$. Since we are only interested in
the ground state and scaling of the energy gap, we can assume without loss of
generality that each $P_X$ is a projector. We also assume that $H$ is
frustration-free, meaning that the ground states minimize the energy of each
term of the Hamiltonian separately, i.e., $P_X\ket{\psi_0} = \ket{\psi_0}$.
Then, the ground space $\mathcal{C}$ is the image of the {\em code projector} $P
= \prod_X P_X$ (henceforth, the $P_X = 0$ are not included in such products). 

This family of lattice models, called {\em local commuting projector code}
(LCPC) includes most models of topological order,
e.g.~\cite{Kitaev03,LW05,KKR10}. It has been proved that LCPC
have a stable spectrum \cite{BHM10, BH,MP} if they obey the following {\em
local topological order} condition.
\begin{defn}
[Local topological order]\label{def:Local-consistency}For any topologically trivial
region $A$, let $P_{A} = \prod_{X: X\cap A \neq \emptyset} P_X$ be  the product
of projectors that intersect region
$A$.
For a system with {\em local topological order}, the density matrices
$\rho_{A}=\mbox{Tr}_{\bar A}P$
and $\rho_{A}^{\textrm{loc}}=\mbox{Tr}_{\overline A}P_{A}$ have the same kernel and $\rho_A$ is proportional to a projector.
\end{defn}

Our main result is that any system with local topological order has only a
constant energy barrier between ground states, similar to Kitaev's code. In
fact, our proof is inspired by the behaviour of Kitaev's code: we exhibit a
procedure creating a point-like defect at one boundary of the lattice and
draging it across to another boundary, at a constant energy cost. The difficulty
stems from the fact that, as noted by Haah and Preskill~\cite{HP12}, it is
unclear whether this process can be realized unitarily in general.

The rest of the article
is organized as follows. We first introduce a few
definitions and review known facts about LCPCs. We then present a local noise
model that requires a constant amount of energy and demonstrate that (i) it can
be executed in a time that scales linearly with system size for any
system with local topological order, and (ii) it corrupts the information
content of the memory. We end with a discussion and general conclusions.

\smallskip
\noindent{\em Background---} Characterizing the thermal stability of a memory
requires detailed knowledge of its thermalization process. Since we seek to
address a broad class of systems, our analysis cannot be model specific. We thus
retain only two essential features common to all thermalization processes: (i)
the bath interacts locally with the system, and (ii) high-energy states are
penalized. As we now explain, we can combine these features to obtain a
sufficient condition for thermal stability.

We will say that a memory with Hamiltonian \eq H has an energy barrier at most $\Delta$ if there exists a ground state $\psi_0$ and a sequence of $T \in {\rm poly}(N)$ CPTP maps $\cE_k$, each acting locally on the system and a finite-dimensional ancilla $A$, such that (i) starting from $\psi_0$, the sequence returns the system to the ground space, (ii) in a state that differs from the initial state $\psi_0$,  and (iii)  the energy  of any intermediate state is at most $\Delta$ above the ground state energy $E_0$. More formally, these conditions are
\begin{align}
\Tr\left[P\cdot \cE_T\ldots \cE_2\cE_1(\psi_0\otimes \rho_A)  \right] \geq \frac 23 \\
\Tr\left[ \psi_0 \cdot \cE_T\ldots \cE_2\cE_1(\psi_0\otimes \rho_A )\right] \leq \frac 13 \\
 \max_{k\in 1,2,\ldots,T} \Tr \left[H \cdot \cE_k\ldots\cE_2\cE_1(\psi_0\otimes\rho_A)\right] - E_0 = \Delta
\end{align}
where $P$ is the code projector, i.e. projector onto the ground space of $H$, and the factors $\frac 23$ and $\frac 13$ are arbitrarily chosen constants. The additional ancillary system, initially in state $\rho_A$, is used to model some finite non-Markovian effects of the bath, so each map $\cE_k$ has complete access to it. The energy barrier of a memory is taken to be the smallest value of $\Delta$ over all such sequences of maps. If a memory has a macroscopic energy barrier $\Delta \geq N^\alpha$ for some constant $\alpha>0$, then any short sequence of local transformations that returns the system to an altered ground state must visit a high energy state, and is therefore thermally stable.  Our main result is obtained by exhibiting a sequence of maps $\cE_k$ with an energy barrier $\Delta$ that is a constant, independent of the system size $N$. In addition, we show that the length $T$ of this sequence is proportional to the linear size of the lattice when the system has local topological order.

We call {\em logical operator} an operator $L$ that maps the ground space to itself, i.e. $[L,P] = 0$ or equivalently $[L,P_X] = 0$ for all $X$.  In
particular, we are interested in logical operators that act non-trivially on the
code space, i.e., $LP \neq P$, as they can alter the encoded information. In a
series of paper~\cite{BT09,BPT10,HP12,KC08}, it was shown that LCPCs
always admit at
least one non-trivial logical
operator supported only on a 1D (constant width) strip of
the lattice. 

An important subclass of LCPCs are stabilizer codes \cite{Gottesman97},
for which $P_{X} = \frac 12(\id+S_X)$ with $S_X$ tensor-product of Pauli
matrices $\sigma_{0,1,2,3}$ (with $\sigma_0$ being the identity $\id$).
Because of this particular structure, the non-trivial string-like
logical operator described above is also a tensor product of Pauli matrices, $L =
\bigotimes_k^\ell \sigma^k_{j_k}$ where $k$ labels the $\ell$ sites along the
strip in some natural way, from left to right, say. 
Then, applying the error sequence
$\left\{\sigma_{j_k}^k\right\}$ will build up to the 
operator $L$, and will only visit intermediate states with a constant energy
above the ground state. Indeed, at an intermediate state
$0<n<\ell$, only a segment $\bigotimes_k^n \sigma^k_{j_k}$ of the
logical operator $L$ has been applied. This segment commutes with all terms
$P_X$ except the ones within distance $\cO(w)$ from site $k$, so only these
terms contribute to the energy: the excitations are point-like objects located
in the vicinity on the end of the string segment, so the memory is unstable \cite{KC08,BT09,Y11b,BDP12a}. This simple argument fails for
more general LCPCs because logical operators do not have a tensor product
structure.

\smallskip
\noindent\textit{Noise model}--- In this Section we present an error sequence $\left\{\cE_k\right\}$ that
achieves a constant energy barrier. To motivate our construction, let us first
consider a simplified version. In a first step, all particles on the strip
are removed and replaced by particles in random states. Mathematically, this action can be described by
the maximally depolarizing channel $\mathcal{D}_k$ on all the sites $k$ of the strip, which is on average equivalent to applying a Haar-random unitary operator on each site $k$,
\begin{equation}
\mathcal{D}_{k}\left[\psi\right]\equiv\mbox{Tr}_{k}\left[\psi\right]\otimes
\id_k/D=\int
U_{k}\proj{\psi}U_{k}^{\dagger}dU_{k}\mbox{.}
\label{eq:depolarizing-random-unitaries}
\end{equation}
At this point, it is clear (via no-cloning) that the information content of the
strip has been wiped out, and that it cannot be recovered by any subsequent
operation. To return the system to its ground space, we can measure the projector $P$, and, with lots of luck, obtain the
result $+1$. Of course this outcome is extremely unlikely, and for this reason we
will refer to this noise model as the {\em fortuitous model}. 
Note that prior to the measurement of $P$, we can apply any CPTP map on the
strip without altering our main conclusion that the information has been erased.

Clearly, the fortuitous model is not what we are seeking. Not only does it
rely on an extraordinary amount of luck, but the state after the first step has
energy proportional to the strip length. We will now show how to realize this
model sequentially to avoid both of these problems. To simplify the
presentation, we coarse-grain the lattice---i.e., we partition the lattice into
balls of radius $w$ and view each ball as sites corresponding to a single
$D$-level spin with $D = d^{\cO(w^2)}$---so we can assume without lost of
generality that: $\Lambda$ is a regular $\ell \times \ell$ square lattice; the
non-zero terms $P_X$ in \eq H act only on $2\times 2$ cells; and there exists a
non-trivial logical operator supported on a single line $\cL$ of the lattice. 
Projectors whose support intersect $\cL$ define the strip
projector $P_{\cL}=\prod_{X\cap \cL\neq\emptyset}P_{X}$
supported on the extended strip $\cL'$, see \fig{coarse-grained-strip}. Similarly,
projectors whose support intersect sites $k-1$ and $k$ on $\cL$ define
\emph{local constraints}  $P_{k-1,k}=\prod_{X\cap \{k-1,k\} \neq\emptyset}P_{X}$.

\begin{figure}
\begin{centering}
\includegraphics[width=0.9\columnwidth]{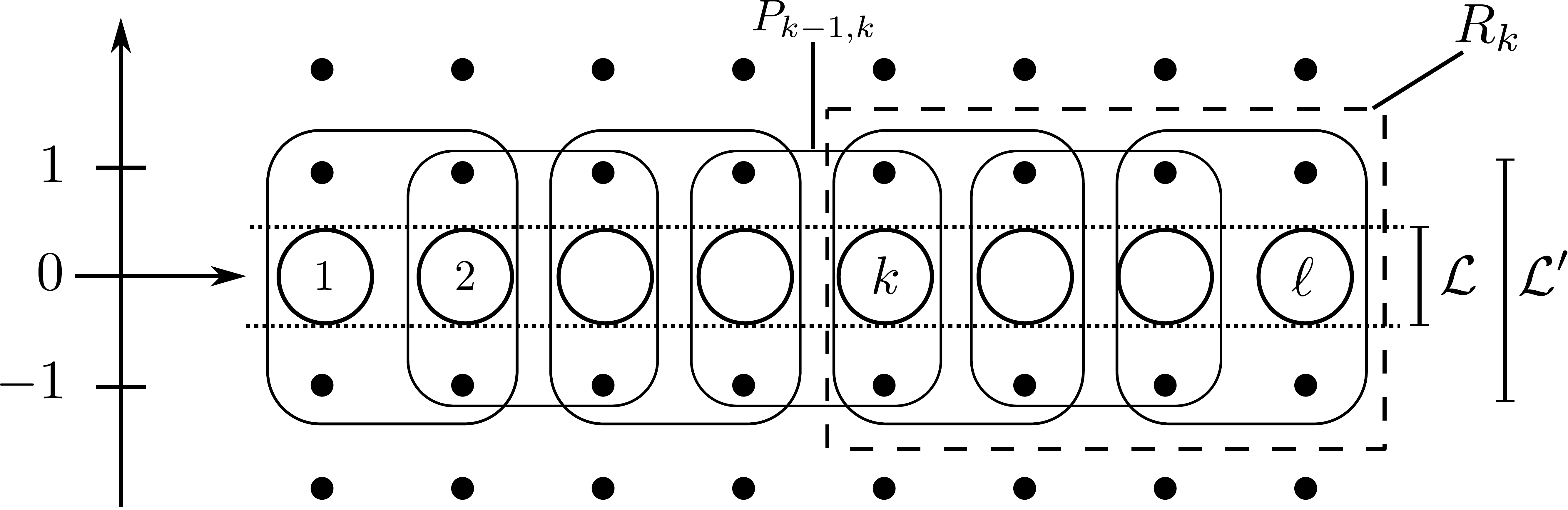}
\par\end{centering}
\caption{\label{fig:coarse-grained-strip}The strip $\cL$ contains $\ell$ sites
(large circles) whose Cartesian coordinates are $\{(k,0)\:1\leq k\leq \ell\}$.
Local constraints $P_{k-1,k}$ act on nearest-neighbours sites
$k-1$ and $k$ and on particles in the extended strip $\cL'=\{(i,j)\,\left|j\right|\leq1\}$. }
\end{figure}

The {\em sequential noise model} interleaves the depolarizing and
projection steps of the fortuitous model to obtain individual \emph{iterations}
for every site $k\in\mathcal{L}$,
which succeed with high probability. Each iteration
consists of several \emph{trials}. Trial $m$ of iteration $k$ corresponds
to (i) applying a \emph{trial unitary} $U_{k}^{(m)}$ on site $k$,
chosen at random from the Haar measure, and (ii) measuring the local
constraint $P_{k-1,k}$. Trials are repeated until a \emph{successful
trial} in which the $+1$ outcome of $P_{k-1,k}$ is obtained and
the next iteration begins. Given the state on the strip, a unitary is
\emph{eligible} if it leads to a successful trial with non-zero probability. 
The initial iteration $k=1$ differs since the constraint
is not measured. Physically, the whole procedure corresponds to
creating a random excitation at iteration 1, and moving it along the strip
across to the opposite edge by
subsequent iterations. 

The sequential noise model only creates intermediate states
of constant energy. The reason is the same as for stabilizer codes: the excitations are point-like objects. Indeed, during iteration $k$, the state is almost everywhere indistinguishable from a ground state because it obeys all constraints $P_{i-1,i}$, except $P_{k-1,k}$ and $P_{k,k+1}$ since only those potentially do not commute with $U_k^{(m)}$. Furthermore, a failed trial during iteration $k$
does not affect the outcome of previous iterations since local constraints commute.
Thus, we only need to show that the expected number of trials
at each iteration is independent of lattice size---so the total duration of the noise process grows linearly with the lattice size $\ell$--- and that the average
effect of a complete sequence of $\ell$ successful iterations has exactly the same effect as the fortuitous model. 

\smallskip
\noindent\textit{Expected number of trials}---
The sequential model would run into a dead-end, an iteration requiring an infinite number of trials, if the
state of the strip admits no eligible unitary at the $k$\textsuperscript{th}
iteration %
\footnote{Note that while the eligibility of a unitary depends on the state
in general, it cannot change if the state has only support on the
kernel of $P_{k-1,k}$.%
}. Such a dead-end occurs in the Ising-like toric code introduced in
\cite{BHM10}, where the plaquette operators $B_{p}$ of the toric
code are replaced by Ising-like interaction $B_{p}B_{q}$ whose symmetry is
broken by a single defect plaquette $B_{p^{*}}$ to recover the toric code ground state. The sequential model could start preparing the $B_{p}=-1$ sector and reach a
dead-end when it encounters the $B_{p^{*}}=+1$ constraint. However,
that code does not have local topological order, and its spectrum is indeed unstable ($B_{p^*}$ is a local order parameter). We now show that such dead-ends do not occur with local topological order.
\begin{prop}
Local topological order implies that, at any iteration $k$, there exists
an eligible unitary $U_{k}$.
\label{deadend}
\end{prop}
\begin{proof}
We will prove the contrapositive. Let $\psi$ be the state during the
$k\textsuperscript{th}$ iteration
for which no $U_{k}$ is eligible. We have 
\begin{equation}
P_{k-1,k}U_{k}\ket{\psi}=0 \quad \forall U_{k},
\end{equation}
whose average over the Haar measure, using Eq. \eqref{eq:depolarizing-random-unitaries},
is \begin{equation}
P_{k-1,k}\left(\mbox{\mbox{Tr}}_{k}\left[\psi\right]\otimes\id_{k}/D\right)=0\mbox{.}\label{eq:constraint-violation}\end{equation}
Tracing out the region $R_{k}=\{(i,j): i\geq k\;\left|j\right|\leq1\}\subset \cL'$
of the extended strip located at the right of site $k$, c.f. \fig{coarse-grained-strip}, Eq. \eqref{eq:constraint-violation}
yields
\begin{equation}
\mbox{Tr}_{k}\left[P_{k-1,k}\right]\mbox{Tr}_{R_{k}}\left[\psi\right]=0\mbox{.}
\end{equation}
Thus, there exists a state $\ket{\xi}$ in the support of $\mbox{Tr}_{R_{k}}\left[\psi\right]$
which is in the image of $P_{i-1,i}$ for $i<k$ but also is in the
kernel of $\mbox{Tr}_{k}\left[P_{k-1,k}\right]$. This entails violation
of local topological order on site $k-2$ since $\mbox{Tr}_{\overline{k-2}}[\xi]$ is in
the kernel of $\rho_{k-2}=\mbox{Tr}_{\overline{k-2}}P$  but
in the image of $\rho_{k-2}^{\textrm{loc}}=\mbox{Tr}_{\overline{k-2}}\left[P_{k-3,k-2}P_{k-2,k-1}\right]$.
\end{proof}

\begin{prop}
When the system has local topological order, the expected number of trials $A_{k}$ at iteration $k$ is finite and independent of the system size.
\end{prop}
\begin{proof}
We introduce two maps
\begin{eqnarray}
\mathcal{P}_{k-1,k}[\rho] & = & P_{k-1,k}\rho
P_{k,k+1}\label{eq:superprojector-satisfaction}\\
\mathcal{Q}_{k-1,k}[\rho] & = &
(\id-P_{k-1,k})\rho(\id-P_{k-1,k})\label{eq:superprojector-violation}
\end{eqnarray}
which represent a successful and failed measurement of the local constraint
$P_{k-1,k}$. 
In a failed trial, the map $\cQ_{k-1,k}$ is always immediately preceded and
followed by a depolarization of site $k$. This sequence can be rewritten in an
equivalent form, c.f. \fig{shuffling-op}
\begin{equation}
\mathcal{E}_{k-1}\otimes\mathcal{D}_{k}=\mathcal{D}_{k}\mathcal{Q}_{k-1,k}\mathcal{D}_{k}
\label{eq:shuffling}
\end{equation}
which defines a \emph{biasing map} $\cE_{k-1}$. This map is not trace-preserving
since the trace of its unnormalized output state is the average probability
of a failed trial.

\begin{figure}
\begin{centering}
\includegraphics[width=0.9\columnwidth]{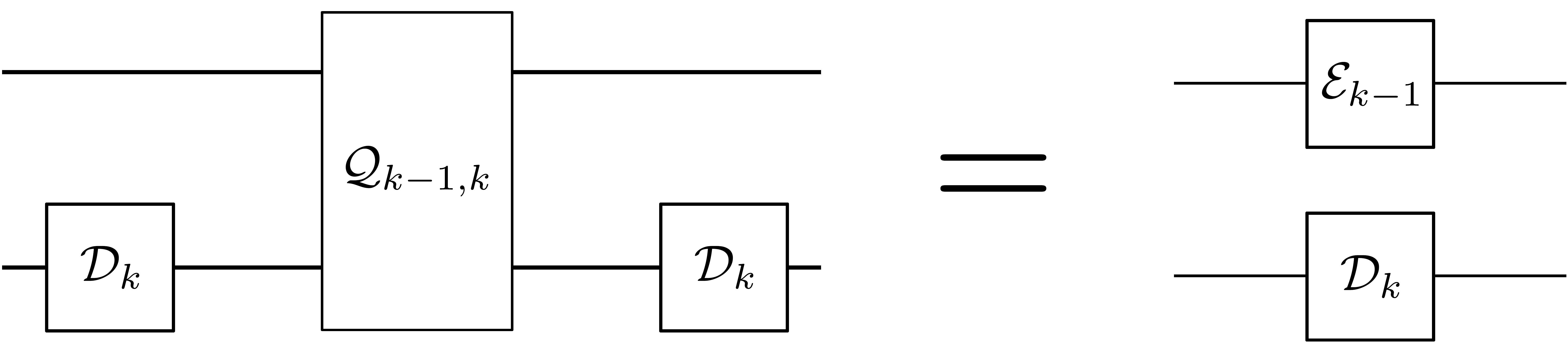}
\par\end{centering}
\caption{\label{fig:shuffling-op}Definition of the biasing map $\mathcal{E}_{k-1}$.
It is only defined on the image of $\mathcal{D}_{k}$, a projection operation.}
\end{figure}

The sequence of $m$ failed trials followed by a successful trial produces the
map 
\begin{equation}
\mathcal{P}_{k-1,k}\mathcal{D}_{k}\left(\mathcal{Q}_{k-1,k}\mathcal{D}_{k}\right)^{m}=\mathcal{P}_{k-1,k}\left(\mathcal{E}_{k-1}^{m}\otimes\mathcal{D}_{k}\right)
\label{eq:success-after-m}
\end{equation}
where we have used Eq. \eqref{eq:shuffling} and $\mathcal{D}_{k}^{2}=\mathcal{D}_{k}$.
Thus, given a state $\psi$, the average probability $p_{k}^{(m)}(\psi)$
of a success after $m$ failures is
$p_{k}^{(m)}(\psi)=\mbox{Tr}\left[\mathcal{P}_{k-1,k}\left(\mathcal{E}_{k-1}^{m}
\otimes\mathcal{D}_{k}\right)\left[\psi\right]\right]$. 
Therefore, the expected number of trials 
\begin{align}
A_{k}(\psi)= & \sum_{m=1}^{\infty}(m+1)
\mbox{Tr}\left[\mathcal{P}_{k-1,k}\left(\mathcal{E}_{k-1}^{m}\otimes\mathcal{D}_
{k}\right)\left[\psi\right]\right] \\
 = &
\mbox{Tr}\left[\mathcal{P}_{k-1,k}\left(\left(\mathcal{I}_{k-1}-\mathcal{E}_{k-1
} \right)^{-2}\otimes\mathcal{D}_{k}\right)\left[\psi\right]\right]
\label{eq:average}
\end{align}
is bounded by the norm of
the superoperator inside the trace
and thus only depends on the microscopic details of $H$, not on system
size.
Note that the geometric sum of \eq{average} converges since $\cE_{k-1}$ cannot
have $+1$ eigenvectors in the groundspace for a local
topological ordered system, since those would
contradict Proposition \ref{deadend}.
\end{proof}

\smallskip
\noindent\textit{Equivalence between models---}We now prove the equivalence
of the sequential and fortuitous models. The effect of iteration $k$ averages to
the map
\begin{eqnarray}
\mathcal{K}_{k-1,k} & = & \sum\nolimits _{m=0}^{\infty}\mathcal{P}_{k-1,k}\left(\mathcal{E}_{k-1}^{m}\otimes\mathcal{D}_{k}\right)\\
 & = & \mathcal{P}_{k-1,k}\left(\left(\mathcal{I}-\mathcal{E}_{k-1}\right)^{-1}\otimes\mathcal{D}_{k}\right)\mbox{,}
 \label{eq:site-randomization}
 \end{eqnarray}
so the total action of the sequential noise model is
\begin{equation}
\mathcal{K}=\prod\nolimits _{k=2}^{L}\mathcal{P}_{k-1,k}\left(\left(\mathcal{I}-\mathcal{E}_{k-1}\right)^{-1}\otimes\mathcal{D}_{k}\right)\mathcal{D}_{1}.
\label{eq:sequential}
\end{equation}
Terms with non-overlapping support trivially commute. We thus move
all depolarizing channels to act first, which globally depolarizes
the strip. To move the biasing operators past the projectors, it suffices to prove that $C\equiv\left[\mathcal{E}_{k},\mathcal{P}_{k-1,k}\right]\mathcal{D}_{k+1}$
is zero. Because of their non-overlapping supports, $\mathcal{D}_{k+1}$
commutes with $\mathcal{P}_{k-1,k}$ and $C=\left[\mathcal{E}_{k}\mathcal{D}_{k+1},\mathcal{P}_{k-1,k}\right]=\left[\mathcal{D}_{k+1}\mathcal{Q}_{k,k+1}\mathcal{D}_{k+1},\mathcal{P}_{k-1,k}\right] = 0$
 since $\mathcal{Q}_{k,k+1}$ and $\mathcal{P}_{k-1,k}$ commute. Hence, the terms of \eq{sequential} can be reordered into 
\begin{equation}
\mathcal{K}=\prod_{k=2}^{L}\mathcal{P}_{k-1,k}\prod_{k=2}^{L+1}\left(\mathcal{I}-\mathcal{E}_{k-1}\right)^{-1}\prod_{k=1}^{L}\mathcal{D}_{k}
\end{equation}
which has the form of the fortuitous model: complete depolarization of the strip, arbitrary transformation on the strip, and projection onto the ground space.

\textit{Discussion---} While we have focused on systems with open boundary
conditions, our results extend straightforwardly to periodic boundaries. The
noise process could begin at some arbitrary location where it would create a
pair of defects and, using the techniques we presented, wrap one of these
defects around the system. A final measurement would then attempt to fuse both
defects back into the ground space. Note that it may be necessary to wrap one
defect several times around the system to obtain a successful fusion. For
instance, Kitaev's code with a twisted boundary requires two complete wraps.

The noise process we presented corrupts the memory after a time that grows
proportionally to the system size, which can be interpreted as a (macroscopic)
upper bound to the storage time. However, there are good reasons to believe that
the actual storage time is in fact independent of the system size. At nonzero
temperature we expect a finite density of defects, so the noise process we
described could be happening in parallel all over the lattice. As pairs of
defects meet, they can fuse to the vacuum with some probability to create longer error strings. The memory time is
then related to the percolation of these error chains, which should be
independent of the system size. 

Our result does not completely close the door to the existence of a robust quantum memory in 2D. First, local topological order is a sufficient, but perhaps not necessary condition for spectral stability. Systems with a stable spectrum that do not have local topological order would escape our conclusions. 
Second,   we have restricted the form of the Hamiltonian. In realistc physical
systems, the terms $P_X$ need not to commute, the ground space can be
frustrated, and interaction can decay algebraically with the distance between
sites. Third, a macroscopic energy barrier is one mechanism that leads to
thermal stability, but they may exist other mechanism. In particular, a system
in contact with a heat bath tends to minimize its free energy $F = E-TS$. Thus,
we could imagine a system with a large entropy barrier: among all possible local
noise sequences, only a vanishingly small fraction can take one ground state to
another, while the overwhelming  majority lead to dead-ends as described above.
Such topological spin-glasses \cite{CC12} could offer an enhanced quantum memory
lifetime. This proposal is distinct from existing studies showing that disorder
induces an exponential localization of anyons \cite{TOC11,WP11,SPI+11,RWH+12},
as those only address zero-temperature storage.

\smallskip
\noindent\textit{Conclusion---} Our main result hints at a general trade-off in
2D between a quantum memory's ability to suppress thermal noise and
its stability to static perturbations. Recent discoveries
\cite{Haah11,BH11,Michnicki} show that this tradeoff is not necessary in
3D. Our result extends prior findings \cite{KC08,BT09} derived for
stabilizer codes to a broader, widely studied class of models that includes
quantum double \cite{Kitaev03}, Levin-Wen \cite{LW05}, and Turaev-Viro
\cite{KKR10} models among others.

\smallskip
\noindent\textit{Acknowledgements---} 
We thank Jeonwang Haah and John Preskill for insightful discussions.  
OLC is funded by NSERC through a Vanier Scholarship. 
This work was supported by Intelligence Advanced Research Projects Activity (IARPA) via Department of Interior National Business Center contract D11PC20167. The U.S. Government is authorized to reproduce and distribute reprints for Governmental purposes notwithstanding any copyright annotation thereon. Disclaimer: The views and conclusions contained herein are those of the authors and should not be interpreted as necessarily representing the official policies or endorsements, either expressed or implied, of IARPA, DoI/NBC, or the U.S. Government.

\bibliographystyle{apsrev4-1}
\bibliography{nogo2D}

\end{document}